\documentclass[11pt]{article}

\input{style}

\declareperson{June}
\declareperson{Nima}

\usepackage{algorithm2e}

\newcommand*{\rows}[1]{#1_{\mathrm{row}}}
\newcommand*{\cols}[1]{#1_{\mathrm{col}}}
\DeclareMathOperator{\detlb}{detlb}
\DeclareMathOperator{\maxdet}{maxdet}
\newcommand*{\I}{\mathcal{I}}
\newcommand*{\cN}{\mathcal{N}}
\newcommand*{\cE}{\mathcal{E}}
\newcommand*{\mU}{\overline{U}}

\title{An Extension of Pl\"ucker Relations with Applications to Subdeterminant Maximization}
\author{Nima Anari}
\author{Thuy-Duong Vuong}
\affil{Stanford University, \textsf{\{anari,tdvuong\}@stanford.edu}}
\date{}

\begin{document}
    \maketitle
    
    \begin{abstract}
    	Given a matrix $A$ and $k\geq 0$, we study the problem of finding the $k\times k$ submatrix of $A$ with the maximum determinant in absolute value. This problem is motivated by the question of computing the determinant-based lower bound of \textcite{LSV86} on hereditary discrepancy, which was later shown to be an approximate upper bound as well \cite{Mat13}. The special case where $k$ coincides with one of the dimensions of $A$ has been extensively studied. \Textcite{Nik15} gave a $2^{O(k)}$-approximation algorithm for this special case, matching known lower bounds; he also raised as an open problem the question of designing approximation algorithms for the general case.
    	
    	We make progress towards answering this question by giving the first efficient approximation algorithm for general $k\times k$ subdeterminant maximization with an approximation ratio that depends only on $k$. Our algorithm finds a $k^{O(k)}$-approximate solution by performing a simple local search. Our main technical contribution, enabling the analysis of the approximation ratio, is an extension of Pl\"ucker relations for the Grassmannian, which may be of independent interest; Pl\"ucker relations are quadratic polynomial equations involving the set of $k\times k$ subdeterminants of a $k\times n$ matrix. We find an extension of these relations to $k\times k$ subdeterminants of general $m\times n$ matrices.
    \end{abstract}

    \section{Introduction}
    \label{sec:intro}
    
    We consider the problem of finding the $k\times k$ submatrix of a given $m\times n$ matrix $A$ that has the largest determinant in absolute value:
    \[ \maxdet_k(A):=\max\set*{\abs{\det(A_{I, J})} \given I\in \binom{[m]}{k}, J\in \binom{[n]}{k}}. \]
    A well-studied special case of this problem asks to find the maximum absolute determinant of a \emph{maximal} submatrix. In other words, $k$ is set to $\min\set{m, n}$. This special case is known in the literature as the largest volume simplex problem or simply (sub)determinant maximization \cite{Kha95, DEFM14, Nik15}, and it was originally framed as the problem of finding a largest simplex in a convex body, a simplex-based analog of the John ellipsoid. The best approximation algorithm for when $k=\min\set{m, n}$ was obtained by \textcite{Nik15} who gave an efficient $2^{O(k)}$-approximation algorithm, improving upon $\log(k)^{O(k)}$-approximation of \cite{DEFM14}, and the earlier $k^{O(k)}$-approximation of \cite{Kha95}, and also matching known lower bounds \cite{DEFM14}.m,
    
    More recently, a line of work has studied various generalizations of the largest volume simplex problem, where the returned indices of the submatrix are required to satisfy a matroid constraint \cite{NS16,AO17,SV17,ESV17,AOV18,MNST20}. This line of work led to fruitful applications in several problems in combinatorial optimization: experimental design, network design, fair allocation, column subset selection, and more \cite[see][for the history and applications]{MNST20}.
    
    Despite the extensive study of variants of the special case $k=\min\set{m, n}$, little has been done for the general case where $k<\min\set{m, n}$. A key motivation behind studying the general case comes from discrepancy theory, namely the problem of computing the determinant lower bound on hereditary discrepancy, due to \textcite{LSV86}. This quantity is defined formally as
    \[ \detlb(A):=\max\set*{\sqrt[k]{\maxdet_k(A)}\given k\geq 0}. \]
    \Textcite{Mat13} showed, by completing earlier results of \textcite{LSV86}, that $\detlb(A)$ is a polylogarithmic approximation to the hereditary discrepancy of $A$. This raised the question of efficiently approximating $\detlb(A)$. \Textcite{NT14} showed how to approximately compute the hereditary discrepancy by bypassing $\detlb(A)$ and instead computing a quantity known as $\gamma_2(A)$; they showed that $\gamma_2(A)$ is a logarithmic approximation of $\detlb(A)$ \cite{MNT20} and a polylogarithmic approximation of hereditary discrepancy. But efficient $O(1)$-approximation of $\detlb$ remains open. \Textcite{Nik15} who obtained the best approximation algorithm for the largest volume simplex problem, posed this as an open problem. Such a result has the potential to improve the approximation factor for hereditary discrepancy, as the worst known gap between $\detlb$ and hereditary discrepancy is only logarithmic \cite[see][]{PD10, Mat13}.
    
    As a step towards answering this question, we show how to approximate $\maxdet_k(A)$ efficiently, with an approximation factor that depends only on $k$.
    \begin{theorem}\label{thm:main}
    	There is a polynomial time algorithm that on input $A\in \R^{m\times n}$, outputs sets of indices $I\in \binom{[m]}{k}$ and $J\in \binom{[n]}{k}$ guaranteeing
    	\[ k^{O(k)}\cdot \abs{\det(A_{I,J})}\geq \maxdet_k(A). \]
    \end{theorem}
    To the best of our knowledge, this is the first nontrivial approximation algorithm for $\maxdet_k$. Our algorithm is based on a simple local search procedure, where in each iteration indices of \emph{up to two} rows and/or columns are replaced by new ones, until an approximate local maximum is found.
    
    Local search and greedy algorithms have been studied for the related problems of largest volume simplex, D-optimal design, and maximum a posteriori inference in (constrained) determinantal point processes \cite{Fed13, KD16, MSTX19, IMOR20}. A key difference in our work, compared to prior works, is that we need to allow \emph{two changes} per iteration. It is easy to construct examples where replacing only one row or one column at a time can get us stuck in an arbitrarily bad local optimum. For example, consider a diagonal matrix:
    \[
    	A:=\begin{bmatrix}
    		d_1 & 0 & \dots & 0\\
    		0 & d_2 & \dots & 0\\
    		\vdots & \vdots & \ddots & \vdots \\
    		0 & 0 & \dots & d_n
    	\end{bmatrix}
    \]
    Any principal $k\times k$ submatrix is a local optimum. Changing any row or column results in a $0$ determinant. But obviously, $d_i$s can be planted in a way that some of the local optima become arbitrarily bad. On the other hand, allowing simultaneous change of a row and a column lets us move between various subsets of $d_i$s, and escape the bad local optima.
    
    \subsection{Techniques}
    
    Despite the simplicity of applying local search to combinatorial optimization problems, it is often difficult to prove approximation guarantees for its performance. We take a page from the study of matroids and discrete convexity \cite{Mur03}, and prove a quantitative exchange inequality for subdeterminants. We will formally show that if $(I, J)$ and $(I^*, J^*)$ are two sets of indices determining $k\times k$ submatrices, one can swap at most two elements in total between $I$ and $I^*$, and  $J$ and $J^*$, and obtain
    \begin{equation}\label{eq:intro-exchange} \abs{\det(A_{I,J})}\cdot \abs{\det(A_{I^*, J^*})}\leq k^{O(1)} \cdot \abs{\det(A_{I\Delta dI, J\Delta dJ})}\cdot \abs{\det(A_{I^*\Delta dI, J\Delta dJ})} \end{equation}
    for $dI\subseteq I\Delta I^*, dJ\subseteq J\Delta J^*$ of total size $\card{dI}+\card{dJ}\in \set{2, 4}$. This can be viewed as a form of discrete log-concavity for the determinant function on submatrices, and allows us to bound the approximation ratio of a local maximum.
    
    Exchange properties have a long history in the theory of matroids, valuated matroids, and M-concavity \cite{Mur03}. Besides their use in proving the performance of greedy and local search algorithms for optimization problems, they have also recently found applications in sampling problems \cite{ALOV20}.
    
    In order to prove the exchange inequality, we find an extension of Pl\"ucker relations to $k\times k$ subdeterminants of $m\times n$ matrices. The relations are in the form of an identity expressing the l.h.s.\ of \cref{eq:intro-exchange} as a linear combination of the possible values, for different choices of $dI, dJ$ on the r.h.s. Classical Pl\"ucker relations establish exactly this form of identity in the case of $k=\min\set{m, n}$, and have been known to be connected to variants of matroids and exchange properties \cite{DW91}, although not quantitative exchanges of the approximate multiplicative type. Our key technical contribution is the establishment of a variant of these identities when $k<\min\set{m, n}$.
    
    Several variants of Pl\"ucker relations have been studied in the literature. For example \textcite{DW91} extended the Pl\"ucker relations to Pfaffians of skew-symmetric matrices. Their extension involves submatrices of varying sizes, and does not immediately yield a relationship involving just $k\times k$ submatrices. Both our approximate exchange inequality, and our extension of Pl\"ucker relations appear to be novel and might be of independent interest. 
    
    \section{Preliminaries}

    \label{sec:prelim}
    We use the notation $[n] = \set{1,\dots,n}$ for integers $n$. We denote the family of subsets of size $k$ from $[n]$ by $\binom{[n]}{k}$. We use $S \Delta T=(S\setminus T)\cup (T\setminus S)$ to denote the symmetric set difference between $S$ and $T$. When $m,n,k$ are clear from context, we denote by $\I$ the family of valid submatrix index pairs for $k\times k$ submatrices
    \[ \I:= \binom{[m]}{k}\times \binom{[n]}{k}. \]

    For a pair $S=(\rows{S}, \cols{S})\in \I$, and a matrix $A\in \R^{m\times n}$, we denote by $A_S=A_{\rows{S},\cols{S}}$ the submatrix of $A$ with rows and columns indexed by $\rows{S}, \cols{S}$ respectively. We extend set operations, such as $\Delta$ to pairs of sets denoting row and column indices in the natural way. For example for $S=(\rows{S}, \cols{S})$ and $U=(\rows{U}, \cols{U})$ we let $S\Delta U=(\rows{S}\Delta\rows{U},\cols{S}\Delta \cols{U})$. Similarly we let $\card{S}=\card{\rows{S}}+\card{\cols{S}}$. The reader might wish to think of pairs of row and column index sets as one single set, with the caveat that row indices are distinguished from column indices.

    Throughout the paper, we keep the input matrix $A \in \R^{m\times n}$ for subdeterminant maximization fixed. We also assume, w.l.o.g.\ that $m\leq n$. For $S=(\rows{S}, \cols{S}) \in \I$, we use $[S]=[\rows{S},\cols{S}]$ and $[A_S]=[A_{\rows{S},\cols{S}}]$, interchangeably as a shorthand for $\det(A_S)=\det(A_{\rows{S},\cols{S}})$.
    
    In \cref{sec:crude}, we use the following famous formula for determinants of rectangular matrix products.
    \begin{fact}[Cauchy-Binet Formula]
    	Let $A\in \R^{m\times n}$ and $B\in \R^{n\times m}$. Then
    	\[ \det(AB)=\sum_{S\in \binom{[n]}{m}} \det(A_{[m], S})\det(B_{S, [m]}). \]
    \end{fact}

    For indices $S=(\rows{S}, \cols{S}), T=(\rows{T}, \cols{T}) \in \I $, let 
    \[ d(S, T): = \card{S\Delta T}/2=\card{\rows{S} \Delta \rows{T}}/2 + \card{\cols{S} \Delta \cols{T}}/2 \]
    be the \textit{distance} between $S$ and $T$.
    
    Armed with this distance, we can define the neighborhoods of a submatrix indexed by $S\in \I$:
    \begin{definition}
    	For $r \geq 0$ let the $r$-neighborhood of $S\in \I$ be 
    	\[ \cN_r(S) := \set{T \in \I \given d(S, T) \leq r }. \]
    \end{definition}
    
    \section{Subdeterminant Maximization via Local Search}\label{sec:local-search}
    In this section we prove our main result, \cref{thm:main}. Our strategy is to use a simple local search that starts with a submatrix indexed by $S\in \I$, and myopically  finds better and better solutions by searching $2$-neighborhoods until no more improvement can be found. 
    
    To make sure that our algorithm terminates within polynomial time, we will only take improvements that increase the magnitude of the determinant by at least a lower multiplicative threshold; for our purposes, even a factor $2$ improvement works. We will then show how to find a good \emph{start}, needed to bound the number of local search steps, by bootstrapping with the help of a crude approximation algorithm.
    
    We will find a locally approximately maximum solution as defined below.
    \begin{definition}
    	For $\alpha > 0,$ we say $S\in \I$ is an $(r,\alpha)$-local maximum if \[ \abs{\det(A_S)} \geq \alpha \abs{\det(A_T)}\] for all $T\in \cN_r(S)$. 
    \end{definition}
    
    \Cref{alg:localsearch} finds this locally approximate maximum. It starts with some arbitrary solution $S_0\in \I$, and iteratively finds $\alpha$-factor improvements within the $2$-neighborhood, until no more improvement can be found.
    \begin{algorithm}
	\SetAlgoLined
	 Let $S\leftarrow S_0$
	 
	 \While{there is $T\in \cN_2(S)$ such that $\alpha \abs{\det(A_T)}> \abs{\det(A_S)}$}{
	   Let $S\leftarrow T$
	 } 
	 Output $S=(\rows{S},\cols{S})$
	 \caption{$\alpha$-Local Search}\label{alg:localsearch}
	\end{algorithm}
    
    It is immediate to see that when \cref{alg:localsearch} terminates, the output is a $(r,\alpha)$-local maximum.
    \begin{proposition}
    	The output of \cref{alg:localsearch} is a $(2,\alpha)$-local maximum.
    \end{proposition}
    The most challenging part of local search algorithms is proving that local (approximate) optimality implies global (approximate) optimality. We appeal to approximate exchange properties that we prove for $k\times k$ subdeterminants, and show the following statement in \cref{sec:exchange}.
    \begin{lemma}\label{lem:local-to-global}
    	Suppose that $S\in \I$ is a $(2, \alpha)$-local maximum. Then $S$ is a $(k/\alpha)^{O(k)}$-approximate global optimum:
    	\[ (k/\alpha)^{O(k)}\cdot \abs{\det(A_S)}\geq \maxdet_k(A). \]
    \end{lemma}
    We prove the remaining part of \cref{thm:main}, that with a suitable choice for $S_0$, \cref{alg:localsearch} runs in polynomial time.
    \begin{proposition}\label{prop:steps-bound}
    	The number of steps taken by \cref{alg:localsearch} starting from $S_0$ is at most
    	\[ \log_{1/\alpha}\parens*{\frac{\maxdet_k(A)}{\abs{\det(A_{S_0})}}}. \]
    \end{proposition}
    \begin{proof}
    	Each iteration improves $\abs{\det(A_S)}$ by a factor of $1/\alpha$. On the other hand, this value can never exceed $\maxdet_k(A)$, and it starts as $\abs{\det(A_{S_0})}$.
    \end{proof}
    In \cref{sec:crude}, we show how to obtain a good $S_0$ by a crude algorithm, that appeals to known results for the case of $k=\min\set{m, n}$. We will formally show the following.
    \begin{lemma}\label{lem:crude}
    	There is a polynomial time algorithm that returns $S_0$ with
    	\[ (n+m)^{O(k)} \cdot \abs{\det(A_{S_0})}\geq \maxdet_k(A). \]
    \end{lemma}
    Having all the ingredients for \cref{thm:main}, we finish its proof.
    \begin{proof}[Proof of \cref{thm:main}]
    	We set $\alpha$ to be some constant below $1$, say $1/2$. We first apply \cref{lem:crude} to obtain a good starting point $S_0$. If $\det(A_{S_0})=0$, then $\maxdet_k(A)=0$, and there is nothing to be done. Otherwise, we run \cref{alg:localsearch} with $\alpha=1/2$. The output of the algorithm, $S$, is a $(2,1/2)$-local maximum, which by \cref{lem:local-to-global}, is a $(2k)^{O(k)}=k^{O(k)}$-approximate solution.
    	
    	Each iteration of \cref{alg:localsearch} clearly runs in polynomial time, since $\cN_2(S)$ has at most $O(k^2(m+n)^2)$ elements. So we just need to bound the number of iterations. But by \cref{lem:crude,prop:steps-bound}, the number of steps is at most
    	\[ \log\parens*{(n+m)^{O(k)}}=O(k \log(m+n)). \]
    \end{proof}
    \begin{remark} The approximation factor of $k^{O(k)}$ is the best possible for local search, even when we consider $(c,\alpha)$-local maxima for any constant number of row/column swaps $c \in \Z_{> 0}$. This is true even for the special case of $k=\min\set{m, n}$. To see why, consider the $\maxdet_k(A)$ problem on input $A\in \R^{k\times 2k}$ defined by the block form
    \[ A = \begin{bmatrix}I_k & c^{-\frac{1}{2}} H_k\end{bmatrix}  \]
    
     where $H_k \in \R^{k\times k}$ is the Hadamard matrix, a matrix with $\pm 1$ entries whose columns are orthogonal to each other. Observe that $A_{[k], [k]} = I_k$ is a $(c,1)$-local maximum, since for any $(I,J) \in \mathcal{N}_c([k],[k])$, after rearranging rows and columns, we can write 
    \[
    	\abs{\det(A_{I,J})} = \abs*{\det\begin{bmatrix} I_{k-c} &  \star \\ 0 & c^{-\frac{1}{2}} D \end{bmatrix}} =\abs{\det (c^{-\frac{1}{2}} D)} \leq 1,
    \]
    where $D \in \{\pm 1\}^{c\times c},$ and $\abs{\det(D)} \leq c^{\frac{c}{2}}$ by the Hadamard inequality. However the global optimum is achieved by the Hadamard matrix part of $A$. Letting $J^* =\set{k+1, \cdots , 2k}$,
    \[ \abs*{\det(A_{[n], J^*) }}  = \abs*{\det(c^{-\frac{1}{2}} H_k)} = \parens*{\frac{k}{c}}^{\frac{k}{2}} \abs*{\det(A_{[k], [k]})}.\]
    In other words, the local optimum is worse than the global optimum by a factor of $(k/c)^{k/2}$.
    \end{remark}
    \section{Approximate Exchange and Local to Global Optimality}\label{sec:exchange}
    
    Here we prove \cref{lem:local-to-global}. Our main tool will be an exchange property, that we state below. First we define the notion of an $r$-exchange.
    
    \begin{definition}\label{def:exchanges}
    	Let $S, T\in \I$ denote two submatrices. We call $U=(\rows{U},\cols{U})$ an $r$-exchange between $S$ and $T$, if $S\Delta U$ and $T\Delta U$ are still indices of $k\times k$ submatrices, $U\subseteq S\Delta T$, and $\card{U}=2r$. Note that $U$ simply represents the exchange of $r$ pairs of rows and/or columns between $S$ and $T$. We denote by $\cE(S, T)$, the set of all $1$-exchanges and $2$-exchanges between $S$ and $T$.
    \end{definition}
    
    Now we are ready to state the key ingredient for proving local to global optimality.
    
    \begin{theorem}[Exchange Property]\label{thm:exchange-property}
    Let $S, T\in \I$ be indices of two $k\times k$ submatrices, and assume that $S\neq T$. Then
    \[ \abs{\det(A_S)}\cdot\abs{\det(A_T)}\leq O(k^2) \max\set*{\abs{\det(A_{S\Delta U})}\cdot\abs{\det(A_{T\Delta U})}\given U\in \cE(S, T)}. \]
    \end{theorem}
    
    Note that \cref{thm:exchange-property} can be thought of a form of discrete log-concavity for subdeterminants. Starting from submatrices $S, T$, we move to two ``nearby'' submatrices $S\Delta U$ and $T\Delta U$ that are closer to $T$ and $S$ respectively, and then we get that up to some error terms, the average log of the determinant goes up.
    
     We will prove \cref{thm:exchange-property} in \cref{sec:plucker} by appealing to a new extension of Pl\"ucker relations, which is an identity between subdeterminants. Here we show how to leverage \cref{thm:exchange-property} to show global approximate optimality from local approximate optimality. Our strategy is to start from $S$ being the locally optimal solution and $T$ being the globally optimal solution, and to gradually move from $T$ to $S$, accumulating at most a $(k/\alpha)^{O(k)}$ loss.
    
	\begin{proof}[Proof of \cref{lem:local-to-global} using \cref{thm:exchange-property}]
	Let $S\in \I$ be a $(2,\alpha)$-local maximum and let $L\in \I$ be the indices of a submatrix that has the highest subdeterminant in magnitude. We first prove the following claim.
	\begin{claim} \label{clm:reduceDistance}
	For any $T \in \I$, 
	there exists $W \in \I$ such that $d(S, W) \leq \max(0, d(S, T) -1)$ and
	\[ 
		\abs{\det(A_T )} \leq O(k^2/\alpha)\cdot \abs{\det(A_W)}.
	\]
	\end{claim}
	\begin{proof}[Proof of \cref{clm:reduceDistance}]
	If $T=S$ then the claim is trivially true, since we can take $W=S$. Assume  $T\neq S$.
	By \cref{thm:exchange-property}, there exists $U \in \cE(S, T)$ such that
	\begin{align*}
	    \abs{\det(A_S)} \cdot \abs{\det(A_T)}
	    &\leq O(k^2)   \cdot \abs{\det(A_{S\Delta U})} \cdot \abs{\det(A_{T\Delta U})} \\
	    &\leq O(k^2)\cdot  \frac{\abs{\det(A_S)}}{\alpha} \cdot  \abs{\det(A_{T\Delta U})}
	\end{align*}
	   where the last inequality follows from the definition of $(2,\alpha)$-local maximum.
	   
	   Setting $W =  T\Delta U $ and dividing both sides by $\abs{\det(A_S)}$ gives the desired inequality. 
	\end{proof}
	Note that initially $d(S, L) \leq 2k$. We can iteratively apply \cref{clm:reduceDistance} for up to $2k$ times, and obtain $W \in \I$ such that  
	$\abs{\det(A_L)} \leq O(k^2/\alpha)^{2k} \abs{\det(A_W)} $ 
	and 
	\[ d(S, W) \leq \max(0, d(S, L) -2k)=0. \]
	The latter condition implies $S=W$, and we are done. 
	\end{proof}
    
    \section{An Extension of Pl\"ucker Relations}\label{sec:plucker}
    
    In this section, we prove \cref{thm:exchange-property} by proving an extension of the Pl\"ucker relations. These are identities relating the $k\times k$ subdeterminants of a matrix. \Cref{thm:exchange-property} will be derived from applying the triangle inequality to these identities.
    
    To give some intuition, let us demonstrate why the regular Pl\"ucker relations, imply an exchange property when $k=\min\set{m, n}$;
    \subsection{Regular Pl\"ucker Relations and Exchange}
    W.l.o.g., let us take $k=m$ and assume $n\geq m$. Given any subsets $S, T\in \binom{[n]}{m},$ the classical Pl\"ucker relation \cite[see, e.g.,][]{DW91} states that, for any fixed $j\in T \setminus S$
	\begin{equation*}
	    \det(A_{[m], S}) \det(A_{[m],T}) = \sum_{i \in S \setminus T}\delta^{i}_j  \det(A_{[m], S\Delta \{i,j\} })\cdot \det(A_{[m],T \Delta\{i,j\} }),
	\end{equation*}
	where $\delta^{i}_j \in\{ \pm 1 \}$ is a sign determined by the indices $i$ and $j$. The triangle inequality then implies the following exchange property
	\begin{equation*}
	    \abs{\det(A_{[m], S})}\cdot \abs{\det(A_{[m],T})} \leq k \cdot \max\set*{\abs{\det(A_{[m], S\Delta \{i,j\} })}\cdot \abs{\det(A_{[m],T \Delta\{i,j\} })}\given i\in S\setminus T, j\in T\setminus S}.
	\end{equation*}
	This is an analog of \cref{thm:exchange-property}, but with just one exchange between $S$ and $T$. As we saw before, we cannot hope for just \emph{one exchange} in the general case of $k<\min\set{m, n}$. But we manage to prove an extended form of Pl\"ucker relations and, by appealing to the triangle inequality, prove \cref{thm:exchange-property}.
	
	\subsection{Extended Pl\"ucker Relations}\label{subsec:plucker}
	
    In this subsection, we state and prove a ``two-dimensional'' extension of Pl\"ucker relations. In trying to find this relationship, we did a bit of guesswork; we knew we were looking for an identity involving only neighbors of the submatrices $S$ and $T$, to make sure we can extract an exchange inequality. By running computer algebra systems on small values of $k$, we discovered the correct form of the identity, and then proceeded to prove it.
    
    Consider $S=(\rows{S}, \cols{S}), T=(\rows{T}, \cols{T}) \in \I$. Note that only the entries in $A_{S \cup T}$ matter and that permuting the rows and/or columns in $S\cup T$ will preserve determinants of $k\times k$ minors up to sign.
    
    We first show a Pl\"ucker relation for the case when $S$ and $T$ are disjoint, i.e., $\rows{S}\cap\rows{T}=\cols{S}\cap\cols{T}=\emptyset$.  
    W.l.o.g., we can assume that
    \begin{equation}\label{eq:positionDisjointCase}
        \rows{S} = \cols{S}= \set{1,\dots,k}\quad\text{and}\quad  \rows{T} = \cols{T}= \set{k+1, \cdots, 2k},
    \end{equation}
    and that $A$ has the following block form:
\[ A=
\left[
\begin{array}{c|c}
C & V \\
\hline
U & D
\end{array}.
\right]
\]
Note that $A_S=C$ and $A_T=D$.

We adopt a few notations for this section.
\begin{itemize}
	\item We use $[\rows{U}, \cols{U}]$ to denote $\det(A_{\rows{U}, \cols{U}}).$
	\item Matrix entries are denoted by lowercase letter. Submatrices are denoted by uppercase letter. For example, we denote entries of submatrix $C$ by $c_{i,j}$ for $i\in \rows{S}, j \in \cols{S}$.
    \item For a set $L$ and $i\in L$ we use $L-i$ and $L^{-i}$ as short hand for $L\setminus \set{i}$. Let $r_L (i)$ denote the rank of $i$ in $L$, i.e., the number of $i'\in L$ that are smaller than $i$. 
    \item 
    For $\rows{U} \subseteq \rows{S} \Delta \rows{T}, \cols{U} \subseteq \cols{S} \Delta \cols{T}$, let $\delta^{U}  = (-1)^{\sum_{i\in \rows{U}} r_*(i) + \sum_{j \in \cols{U}} r_*(j)},$ 
    where, with some abuse of notation we use $r_*$ for both row indices and column indices, and let \[r_*(i) = \begin{cases} r_{\rows{S}}(i)\text{ if $i\in \rows{S}$} \\ r_{\rows{T}}(i)\text{ if $i\in \rows{T}$} \end{cases} ,\quad  r_*(j) = \begin{cases} r_{\cols{S}}(j )\text{ if $j\in \cols{S}$} \\ r_{\cols{T}}(j)\text{ if $j\in \cols{T}$} \end{cases}.\]
\end{itemize}    

\begin{lemma}[Extended Pl\"ucker Relation in the Disjoint Case]  \label{lemma:disjointCase}
Consider $S=(\rows{S},\cols{S}), T=(\rows{T}, \cols{T})$ as in \cref{eq:positionDisjointCase}.

Let $\Omega : = (S, T)$. Define
\begin{equation} \label{eq:siDisjoint}
    \begin{split}
        &s_1(\Omega)=  \sum_{i, j, i', j' } \delta^{\{i, i'\}, \{j,j'\}} [\rows{S} \Delta \{i, i'\}, \cols{S} \Delta \{j,j'\} ] \times [\rows{T} \Delta \{i, i'\}, \cols{T} \Delta \{j,j'\} ]\\
        &s_2(\Omega) = (-1)^k \sum_{i, i'} \delta^{\{i,i'\}, \emptyset}   [\rows{S}\Delta\{i,i'\}, \cols{S}] \times [\rows{T}\Delta\{i,i'\}, \cols{T}]\\
&s_3(\Omega) = \sum_{i<h, i'<h'}\delta^{\{i,i',h,h'\}, \emptyset} [\rows{S}\Delta\{i,h, i', h'\}, \cols{S}]\times [\rows{T}\Delta \{i,h, i', h'\}, \cols{T}], 
    \end{split}
\end{equation}
where in above summations,
$i,h\in \rows{S}, i',h' \in \rows{T}, j \in \cols{S}, j' \in \cols{T}.$ 

Let $s_i: =s_i(\Omega).$ Then, we have the following relation
\begin{equation} \label{eq:pluckerRelation}
    s_1 -  2 (k-1) s_2 - 4 s_3 - k^2 [\rows{S},\cols{S}] [\rows{T}, \cols{T}] = 0
\end{equation}
\end{lemma}

The proof is elementary; we only use well-known identities about the determinant and perform some algebraic manipulation.

\begin{proof}[Proof of \cref{lemma:disjointCase}]
Expanding $[\rows{S} \Delta \{i, i'\}, \cols{S} \Delta \{j,j'\} ]$ along row $i'$, we get:
\[ [\rows{S} \Delta \{i, i'\}, \cols{S} \Delta \{j,j'\} ]   =   d_{i', j'} [\rows{S} -i, \cols{S} - j] + \sum_{\ell \in \cols{S}\setminus \{ j\} } (-1)^{k+r_{\cols{S}^{-j}}(\ell) } u_{i', \ell } [\rows{S}-i, \cols{S} -j-\ell  + j' ]. \]
Expanding $[\rows{S}-i, \cols{S} -j -\ell + j' ]$ along column $j'$, we get:
\[ [\rows{S}-i, \cols{S} -j -\ell + j' ] = \sum_{h \in \rows{S} \setminus \{i\}} (-1)^{k-1 + r_{\rows{S}^{-i}} (h) } v_{k, j'} [\rows{S}-i - h,\cols{S} -j -\ell  ]. \]
Thus
\begin{multline*} [\rows{S} \Delta \{i, i'\}, \cols{S} \Delta \{j,j'\} ] =\\ d_{i', j'} [\rows{S} -i, \cols{S} - j] - \sum_{h,\ell} (-1)^{ r_{\rows{S}^{-i}} (h) + r_{\cols{S}^{-j}}(\ell) } u_{i', \ell } v_{k, j'} [\rows{S}-i-h,\cols{S} -j -\ell ].
\end{multline*}
Similarly,
\begin{multline*}
	 [\rows{T} \Delta \{i, i'\}, \cols{T} \Delta \{j,j'\} ] =\\ c_{i, j} [\rows{T} -i', \cols{T} - j'] - \sum_{h',\ell'} (-1)^{r_{\rows{T}^{-i'}} (h') + r_{\cols{T}^{-j'}}(\ell') } v_{i, \ell' } u_{h', j} [\rows{T}-i' - h',\cols{T} -j' -\ell'  ] \end{multline*}
Now, consider $[\rows{S} \Delta \{i, i'\}, \cols{S} \Delta \{j,j'\} ] \times [\rows{T} \Delta \{i, i'\}, \cols{T} \Delta \{j,j'\} ]$ as a multivariate polynomial $p$ in variables $\vec u = \{u_{\cdot,\cdot}\},\vec v= \{v_{\cdot,\cdot}\}.$ For $s\in \set{0,1,2}$ let $p_s^{i,i', j, j'}$ denote the sum over monomials of $p$ which have degree $s$ in $\vec{u}$ and in $\vec{v}.$ We will omit the superscript when appropriate.

We further decompose $p_1$ into
\begin{align*}
    &p_1 = -(p_{1A} + p_{1B}) \\
    & p_{1A} = \sum_{h, \ell} (-1)^{ r_{\rows{S}^{-i}} (h) + r_{\cols{S}^{-j}}(\ell) } c_{i, j} [\rows{T} -i', \cols{T} - j'] u_{i', \ell } v_{h, j'} [\rows{S}-i-h,\cols{S} -j -\ell  ] \\
    &p_{1B}= \sum_{h', \ell'} (-1)^{r_{\rows{T}^{-i'}} (h') + r_{\cols{T}^{-j'}}(\ell') } d_{i', j'} [\rows{S} -i, \cols{S} - j]  v_{i, \ell' } u_{h', j} [\rows{T}-i' - h',\cols{T} -j' -\ell'  ] )
\end{align*}
\begin{claim} We have
\begin{equation} \label{eq:p1}
    \sum \delta^{\{i,i'\},\{j, j'\} }p_1^{i,i', j, j'} = 2 (k-1) (-1)^k\sum_{h, i'} (-1)^{h+i'}   [\rows{S}-h+i', \cols{S}] \cdot [\rows{T}-i'+h, \cols{T}]
\end{equation}
\end{claim}
\begin{proof}
In $\sum_{i,j} \delta^{\{i,i'\},\{j, j'\}} p_{1A}^{i,i', j, j'}$ we consider the sum of all terms with the same $i',j', h, l$. Note that $r_{\cols{S}^{-j}}(\ell) + r_{\cols{S}}(j) = r_{\cols{S}}(\ell) +r_{\cols{S}^{-\ell}}(j) +1 \mod 2$. This is because, 

$r_{\cols{S}^{-j}}(\ell) = \begin{cases}r_{\cols{S}}(\ell) \text{ if $\ell < j $} \\ r_{\cols{S}}(\ell) -1 \text{ if $\ell> j $} \end{cases}$ and $r_{\cols{S}^{-\ell}}(j) = \begin{cases} r_{\cols{S}}(j)-1 \text{ if $\ell < j $} \\ r_{\cols{S}}(j) \text{ if $\ell> j $} \end{cases}.$

Similarly, $r_{\rows{S}^{-i}}(h) + r_{\rows{S}}(i) = r_{\rows{S}}(h) +r_{\rows{S}^{-h}}(i) +1 \mod 2$.

Thus, this sum is exactly,
\begin{align*}
&\delta^{\{i', h\}, \{j', \ell\} }
u_{i', \ell } v_{h, j'}   [\rows{T} -i', \cols{T} - j']  \sum_{i \neq h,j \neq \ell} (-1)^{ r_{\cols{S}^{-\ell}}(j) + r_{\rows{S}^{-h}}(i)} c_{i, j} [\rows{S}-i - h,\cols{S} -j -\ell  ]  \\
= &(k-1)\delta^{\{i', h\}, \{j', \ell\} } u_{i', \ell } v_{h, j'}   [\rows{T} -i', \cols{T} - j']  [\rows{S}-h, \cols{S} - \ell],
\end{align*}
Indeed, for each $i\in \rows{S}-h,$ expanding $[\rows{S}-h, \cols{S} - \ell]$ along row $i$ gives \[\sum_{j \neq \ell} (-1)^{ r_{\cols{S}^{-\ell}}(j) + r_{\rows{S}^{-h}}(i)} c_{i, j} [\rows{S}-i - h,\cols{S} -j -\ell  ].\] Taking sum over $i\in \rows{S}-h$ gives the above equality.

Thus 
\begin{align*}
&\sum \delta^{\{i,i'\},\{j,j'\}} p_{1A}^{i,i',j,j'} \\
&=  (k-1) \sum_{i', j', h, \ell } \delta^{\{i', h\}, \{j', \ell\} }  u_{i', \ell } v_{h, j'}   [\rows{T} -i', \cols{T} - j']  [\rows{S}-h, \cols{S} - \ell]   \\
&= -(k-1) (-1)^k\sum_{h, i'} \delta^{\{i',h\},\emptyset} \Big((\sum_{\ell} (-1)^{h+r_{\cols{S}}(\ell) } u_{i', \ell} [\rows{S}-h, \cols{S} - \ell])\times\\ & \qquad(\sum_{j' } (-1)^{1+r_{\cols{T}}(j')}  v_{h,j'}[\rows{T} -i', \cols{T} - j'] ) \Big)  \\
&=  -(k-1) (-1)^k\sum_{h, i'} \delta^{\{i',h\},\emptyset}  [\rows{S}-h+i', \cols{S}] [\rows{T}-i'+h, \cols{T}]
\end{align*}
Similarly,
$$\sum p_{1B} = -(k-1) (-1)^k\sum_{i, h'}\delta^{\{i',h\}}   [\rows{S}-i+h', \cols{S}] [\rows{T}-h'+i, \cols{T}] $$
\end{proof}
Next, we show 
\begin{claim}
\begin{equation} \label{eq:p2}
    \sum_{i,j, i', j'} \delta^{\{i,i'\},\{j, j'\}} p_2^{i, i',j, j'}   = 4 \sum_{i<h, i'<h'} [\rows{S}-i-h+i'+h', \cols{S}][\rows{T}-i'-h'+i+h, \cols{T}]
\end{equation}
\end{claim}
\begin{proof}
Recall that
	\begin{align*}
	p_2^{i,i',j,j'} = \sum_{h, \ell, h', \ell'} (-1)^{ \omega(h,h',\ell,\ell') } u_{i',\ell} v_{i,\ell'} v_{h, j'} u_{h', j}[\rows{S}\setminus\{i, h\},\cols{S} \setminus\{j ,\ell\}  ]  [\rows{T}\setminus\{i', h'\},\cols{T} \setminus \{j',\ell'\}  ],
	\end{align*}
    where $\omega(h,h',\ell,\ell')  = r_{\rows{S}^{-i}} (h) + r_{\cols{S}^{-j}}(\ell)  + r_{\rows{T}^{-i'}} (h') + r_{\cols{T}^{-j'}}(\ell') .$ Taking sum and rearranging terms, we have
    \begin{align*}
    &\sum_{i,j, i', j'} \delta^{\{i,i'\},\{j, j'\}} p_2^{i, i',j, j'} 
    = \sum_{i,i',h,h'} (-1)^{r_{\rows{S}}(i)+r_{\rows{T}}(i')+r_{\rows{S}^{-i}} (h)  + r_{\rows{T}^{-i'}} (h')  } X_{i,i',h,h'} \times Y_{i,i',h,h'} 
    \end{align*}
    where
    \begin{align*}
    &X_{i,i',h,h'}  =\sum_{j, \ell} (-1)^{r_{\cols{S}^{-j}} (\ell) + r_{\cols{S}}(j)} u_{i', \ell} u_{h',j} [\rows{S}-i-h, \cols{S} - j - \ell] \\
    & Y_{i,i',h,h'}  =\sum_{j', \ell'} (-1)^{r_{\cols{T}^{-j'}} (\ell') + r_{\cols{T}}(j')} v_{i,\ell'} v_{h, j'} [\rows{T}-i'-h', \cols{T} - j' - \ell'] \\
    \end{align*}
	Expanding $ [\rows{S}-i-h+i'+h', \cols{S}]$ along row $h'$ then $i'$, we get
	\begin{align*}
    [\rows{S}\Delta\{i,i',h,h'\}, \cols{S}] &= \sum_{j}  (-1)^{r_{\cols{S}}(j)+ k+ \1\{i'> h'\}} u_{h', j} [\rows{S}-i-h+i', \cols{S} - j]\\
    &= \sum_{j,\ell} (-1)^{r_{\cols{S}}(j)+ k+ \1\{i'> h'\}} (-1)^{r_{\cols{S}^{-j}} (\ell) + k-1}
     u_{i', \ell} u_{h',j} [\rows{S}-i-h, \cols{S} - j - \ell]\\
    &= (-1)^{\1\{i'< h'\}} X_{i,i',h,h'} \\
    \end{align*}
    Similarly, $Y_{i,i',h,h'}  = (-1)^{\1\{i< h\}} [\rows{T}\Delta\{i,i',h,h'\}, \cols{T}]$. Note that $\1\{i< h\} +  r_{\rows{S}}(i)+r_{\rows{S}^{-i}} (h) \equiv 0\pmod{2}$. A similar equation holds for $i',h'$. Substituting back in we get the desired equation
    \[
    \sum_{i,j, i', j'} \delta^{\{i,i'\},\{j, j'\}} p_2^{i, i',j, j'}  =4\sum_{i<h,i'<h'}  [\rows{S}\Delta\{i,i',h,h'\}, \cols{S}] [\rows{T}\Delta\{i,i',h,h'\}, \cols{T}] \]
\end{proof} 
Lastly, we compute $\sum_{i,j,i',j'} \delta^{\{i,i'\},\{j,j'\}} p_0^{i,j, i', j'}.$ By rearranging terms and using the determinant expansion for  $[\rows{S}, \cols{S}]$ and $[\rows{T}, \cols{T}]$, we get:
\begin{equation}\label{eq:p0}
    \begin{split} 
     &\sum_{i,j,i',j'} \delta^{\{i,i'\},\{j,j'\}} p_0^{i,j, i', j'}\\
        &= \sum \delta^{\{i,i'\},\{j,j'\}} (d_{i', j'} [\rows{S} -i, \cols{S} - j]  c_{i, j} [\rows{T} -i', \cols{T} - j'])\\
    &=(\sum_{i,j} (-1)^{r_{\rows{S}}(i)+r_{\cols{S}}(j)} c_{i, j}  [\rows{S} -i, \cols{S} - j]  )(\sum_{i', j'} (-1)^{r_{\rows{T}}(i')+r_{\cols{T}}(j')} d_{i', j'} [\rows{T} -i', \cols{T} - j'])\\
    &= (t [\rows{S}, \cols{S}])  (t [\rows{T}, \cols{T}])
    \end{split}
\end{equation}

Substituting equations \cref{eq:p0,eq:p1,eq:p2} back into $s_1$ we get \cref{eq:pluckerRelation}.
\end{proof}

Now consider the general case when $\rows{S}, \rows{T}$ and $\cols{S}, \cols{T}$ are not necessarily disjoint. We will create a new larger matrix $A$ with a new set of row and column indices. In particular we create new disjoint subsets $\rows{S}^*, \rows{T}^*$ and $\cols{S}^*, \cols{T}^*$ with copied versions of common rows and columns. We use \cref{lemma:disjointCase} for $\rows{S}^*, \rows{T}^*, \cols{S}^*, \cols{T}^* $, then argue that any nonzero terms in \cref{eq:pluckerRelation} 
must be equal to $[\rows{S} \Delta \rows{U}, \cols{S} \Delta \cols{U}][\rows{T} \Delta \rows{U}, \cols{T} \Delta \cols{U}]$ for some $U \subseteq \rows{S} \Delta \rows{T}, \cols{U} \subseteq \cols{S} \Delta \cols{T}$.

Let $r:=\card{\rows{S} \cap \rows{T}}, c:= \card{\cols{S} \cap \cols{T}}$. W.l.o.g., we can assume 
\begin{equation} \label{eq:positionGeneralCase}
\begin{split}
     &\rows{S} = \set{1, \cdots, r, r+1, \cdots, k}, \rows{T} = \set{1, \cdots, r, k+(r+1), \cdots, 2k},\\ 
     &\cols{S} =  \set{1, \cdots, c, c+1, \cdots, k}, \cols{T} =  \set{1, \cdots, c, k+(c+1), \cdots, 2k}. 
\end{split}
\end{equation}
For $i \in [r]$, set row $k+i$ to be identical to row $i$. For $j\in [c]$,  set column $k+j$ to be identical to row $j$.

Let $\rows{S}^*: = \rows{S}, \cols{S}^*:= \cols{S} , \rows{T}^{*} = \{k+1, \cdots, 2k\}, \cols{T}^* = \{k+1, \cdots, 2k\}.$ Clearly, 
$\rows{S}^* \cap \rows{T}^* = \cols{S}^* \cap\cols{T}^* = \emptyset .$

Let $\Omega^* = (S^*, T^*)$ and $s_i^*: = s_i(\Omega^*)$ as in \cref{eq:siDisjoint}.
We first prove the following claims on the structure of nonzero terms in $s_1^*, s_2^*, s_3^*.$
\begin{claim}\label{clm:structureRow}
Consider $\rows{U}\subseteq \rows{S}^*,  \rows{U}'\subseteq \rows{T}^*$ of the same cardinality. Let $\rows{\mU} = \rows{U} \cup \rows{U}'.$ Consider sets $V, W $ of the same cardinality $t$ that partition $\cols{S}^* \cup\cols{T}^*.$ 

If there exists $i \in \rows{U} \cap [r]$ such that $k+i \not \in \rows{U}'$ then $[\rows{S}^* \Delta \rows{\mU} , V] [\rows{T}^* \Delta \rows{\mU}, W] = 0.$

If there exists $k+i \in \rows{U}' \cap \{k+1, \cdots ,k+r\}$ such that $i \not \in \rows{U}$ then $[\rows{S}^* \Delta \rows{\mU} , V] [\rows{T}^* \Delta \rows{\mU}, W] = 0.$
\end{claim}
\begin{claim}\label{clm:structureColumn}
Consider $\cols{U}\subseteq \cols{S}^*,  \cols{U}'\subseteq \cols{T}^*$ of the same cardinality. Let $\cols{\mU} = \cols{U} \cup \cols{U}'.$ Consider sets $V, W $ of the same cardinality $t$ that partition $\rows{S}^* \cup \rows{T}^*.$ 

If there exists $i \in \cols{U} \cap [c]$ such that $k+i \not \in \cols{U}'$ then $[V, \cols{S}^* \Delta \cols{\mU} ] [W, \cols{T}^* \Delta \cols{\mU}] = 0.$

If there exists $k+i \in \cols{U}' \cap \{k+1, \cdots ,k+c\}$ such that $i \not \in \cols{U}$ then $[V,\cols{S}^* \Delta \cols{\mU} ] [W, \cols{T}^* \Delta \cols{\mU}] = 0.$
\end{claim}
We prove \cref{clm:structureRow}. The argument for \cref{clm:structureColumn} is similar.
\begin{proof}[Proof of \cref{clm:structureRow}]
We prove the first statement. The second one follows by a similar argument, since the role of $\rows{U}, \rows{U}'$ are symmetric.

Suppose there exists $i \in \rows{U} \cap [c]$ such that $k+i \not \in \rows{U}'.$ Then $\rows{T}^* \Delta \rows{\mU}$ contains both rows $i$ and $k+i,$ which are identical by our construction, thus $[\rows{T}^* \Delta \rows{\mU}, W] = 0.$
\end{proof}
\begin{lemma} \label{lemma:generalCase}

Consider $S=(\rows{S},\cols{S}), T=(\rows{T},\cols{T})$ as in \cref{eq:positionGeneralCase}.

Let $\Omega := (S, T), r := \card{\rows{S} \cap \rows{T}}, c := \card{\cols{S} \cap \cols{T}}.$ 

Define
\begin{equation}
    \begin{split}
        &s_1 (\Omega) = (-1)^{r+c}\sum_{i, j, i', j' } \delta^{\{i,i'\},\{j,j'\}} [\rows{S} \Delta \{i, i'\}, \cols{S} \Delta \{j,j'\} ] \times [\rows{T} \Delta \{i, i'\}, \cols{T} \Delta \{j,j'\} ]\\
        &s_2(\Omega) = (-1)^{k-r}\sum_{i, i'} \delta^{\{i,i'\}, \emptyset}   [\rows{S}\Delta\{i,i'\}, \cols{S}] [\rows{T}\Delta\{i,i'\}, \cols{T}]\\
        &\hat{s}_2(\Omega) = (-1)^{k-c} \sum_{j, j'} \delta^{ \emptyset,\{j ,j'\}}   [\rows{S}, \cols{S} \Delta \{j, j'\} ] [\rows{T}, \cols{T} \Delta \{j, j'\}]\\
&s_3(\Omega) = \sum_{i<h, i'<h'} \delta^{\{i,h,i',h'\},\emptyset}[\rows{S}\Delta\{i,h, i', h'\}, \cols{S}][\rows{T}\Delta \{i,h, i', h'\}, \cols{T}] \\
 &\hat{s}_3(\Omega) = \sum_{j < \ell, j' < \ell'} \delta^{ \emptyset,\{j, \ell ,j', \ell'\}}   [\rows{S}, \cols{S} \Delta \{j, \ell, j', \ell'\} ] [\rows{T}, \cols{T} \Delta \{j,\ell, j', \ell'\}],
    \end{split}
\end{equation} 
where in above summations, $j,\ell \in \cols{S} \setminus \cols{T}; j',\ell' \in \cols{T} \setminus \cols{S}; i,h \in \rows{S}\setminus \rows{T}; i',h' \in \rows{T}  \setminus \rows{S}. $

Let $s_i: = s_i(\Omega).$ we have the following relations.
\begin{equation}
     \begin{split}
         (k^2 - 2(k-1) r + 4 \binom{r}{2} - rc) [\rows{S},\cols{S}][\rows{T}, \cols{T}] &=  s_1- r \hat{s}_2  -(  2(k-1) +c -4r ) s_2 - 4 s_3\\
         (k^2 - 2 (k-1)c + 4 \binom{c}{2} - rc) [\rows{S},\cols{S}][\rows{T}, \cols{T}]& =  s_1- c s_2  -(2(k-1) + r-4c) \hat{s}_2 - 4 \hat{s}_3\\
     \end{split}
\end{equation}

Summing the two equations above, we get:
\begin{multline} \label{eq:generalCase}
    ((k-r)^2 + (k-c)^2+(r-c)^2) [\rows{S},\cols{S}][\rows{T}, \cols{T}] =\\ 2 s_1 - 2(k-1+r -2c) \hat{s}_2  - 2(k-1+c -2r) s_2 - 4(s_3 +\hat{s}_3) 
\end{multline}
\end{lemma}
\begin{proof}
We prove the first statement. The second one can be obtained by switching the role of columns and rows. 
Consider $s_1^*$.

Let $X_{i,i',j,j'} := \delta^{\{i,i'\},\{j,j'\}} [\rows{S}^* \Delta \{i, i'\}, \cols{S}^* \Delta \{j,j'\} ] [\rows{T}^* \Delta \{i, i'\}, \cols{T}^* \Delta \{j,j'\} ].$
By \cref{clm:structureRow} and \cref{clm:structureColumn}, any nonzero term $X_{i,i',j,j'}$ in $s_1^*$ must belong to one of the following cases:
\begin{enumerate}
\item $i\in [r], i' = k+i, j\not \in [c], j' \not \in \{k+1, \cdots, k+c\}:$ 
Note that $j \in \cols{S} \setminus \cols{T}, j'\in \cols{T} \setminus \cols{S} .$ 

Obviously $\delta^{\{i,i'\},\{j,j'\}}=\delta^{\emptyset, \{j,j'\}} .$ 

Since rows $k+i$ and $i$ are identical,
\begin{align*}
  [\rows{S}^* \Delta \{i, i'\}, \cols{S}^* \Delta \{j,j'\} ] &= (-1)^{k+i}  [\rows{S}^*, \cols{S}^* \Delta \{j,j'\}] = (-1)^{k+i}  [\rows{S}, \cols{S} \Delta \{j,j'\}]\\ 
  [\rows{T}^* \Delta \{i, i'\}, \cols{T}^* \Delta \{j,j'\} ] &=(-1)^{1+i}  [\rows{T}^*, \cols{T}^* \Delta \{j,j'\} ] = (-1)^{1+i+c}  [\rows{T}, \cols{T} \Delta \{j,j'\} ]
\end{align*}
thus
\[X_{i,i',j,j'} = (-1)^{k-c+1} \delta^{\emptyset, \{j,j'\}} [\rows{S}, \cols{S} \Delta \{j,j'\} ] [\rows{T}, \cols{T} \{j,j'\}].\]
     \item $j\in [c], j' = k+j, i\not \in [r], i' \not \in \{k+1, \cdots, k+r\}:$ similarly, 
     \[X_{i,i',j,j'}= (-1)^{k-r+1} \delta^{\{i,i'\},\emptyset} [\rows{S} \Delta \{i, i'\}, \cols{S}] [\rows{T} \Delta \{i, i'\}, \cols{T}].\]
    \item $i\in [r], j \in [c], i' = k+i, j' = k+j$: 
    $X_{i,i',j,j'} =   [\rows{S}, \cols{S}] [\rows{T}, \cols{T}].$
     \item $i \in \rows{S}\setminus \rows{T}, i' \in \rows{T}\setminus \rows{S}, j \in j \in \cols{S} \setminus \cols{T}, j'\in \cols{T} \setminus \cols{S} ,$ then 
     \begin{align*}
        X_{i,i',j,j'} &= \delta^{\{i,i'\},\{j,j'\}}  [\rows{S}^* \Delta \{i, i'\}, \cols{S}^* \Delta \{j,j'\} ] [\rows{T}^* \Delta \{i, i'\}, \cols{T}^* \Delta \{j,j'\} ] \\
        &= (-1)^{r+c}\delta^{\{i,i'\},\{j,j'\}}  [\rows{S} \Delta \{i, i'\}, \cols{S} \Delta \{j,j'\} ] [\rows{T} \Delta \{i, i'\}, \cols{T} \Delta \{j,j'\} ]
     \end{align*}
\end{enumerate}
We can rewrite
\begin{equation*}
    s_1^* = s_1 + rc [\rows{S},\cols{S}][\rows{T}, \cols{T}] - (r\, \hat{s}_2 + c\, s_2) 
\end{equation*}
By a similar argument
\begin{align*}
    &s_2^* = s_2 - r [\rows{S},\cols{S}][\rows{T}, \cols{T}]\\
    &s_3^* = s_3 + \binom{r}{2} [\rows{S},\cols{S}][\rows{T}, \cols{T}] - r s_2
\end{align*}
Substitute into \cref{eq:pluckerRelation} we get:
\begin{align*}
    (k^2 -2 (k-1) r + 4 \binom{r}{2} - rc) [\rows{S},\cols{S}][\rows{T}, \cols{T}] =  s_1- r \hat{s}_2  -( 2(k-1) +c-4r ) s_2 - 4 s_3
\end{align*}
\end{proof}

\subsection{From the Extended Pl\"ucker Relations to Exchange}
Armed with the extended Pl\"ucker relations, we are now ready to prove \cref{thm:exchange-property}.
\begin{proof}[Proof of \cref{thm:exchange-property}]
We can permute the rows and columns of $A$ so that $\rows{S}, \cols{S}, \rows{T}, \cols{T}$ are as in \cref{eq:positionGeneralCase}, while preserving the absolute value of determinant of minors. W.l.o.g., we can assume the permutation has already been applied; thus \cref{eq:generalCase} holds.

Let $r = \card{\rows{S}\cap \rows{T}}, c = \card{\cols{S} \cap \cols{T}}.$

With $\Omega=(S, T)$, we define 
\begin{equation} \label{eq:exchangeBase}
    \begin{split}
    \mathcal{E}_1^{\Omega}: &= \set*{(\{i,i'\},\{j,j'\}) \given i\in \rows{S}\setminus \rows{T}, i' \in  \rows{T}  \setminus \rows{S} , j  \in \cols{S} \setminus \cols{T}, j'\in \cols{T} \setminus \cols{S} }\\
    \mathcal{E}_2^{\Omega}: &= \set*{(\{i,i'\},\emptyset) \given i\in \rows{S}\setminus \rows{T}, i' \in \rows{T}  \setminus \rows{S}  }\\
    \hat{\mathcal{E}}_2^{\Omega}: &= \set*{(\emptyset,\{j,j'\}) \given   j  \in \cols{S} \setminus \cols{T}, j'\in \cols{T} \setminus \cols{S} }\\
    \mathcal{E}_3^{\Omega}: &= \set*{(\{i,h,i',h'\},\emptyset) \given i,h\in \rows{S}\setminus \rows{T}; i',h' \in  \rows{T}  \setminus \rows{S}  }\\
    \hat{\mathcal{E}}_3^{\Omega}: &= \set*{(\emptyset,\{j, \ell, j', \ell'\}) \given   j, \ell  \in \cols{S} \setminus \cols{T}; j', \ell'\in \cols{T} \setminus \cols{S} }\\
    \end{split}
\end{equation}
then $\mathcal{E}(S, T) = \mathcal{E}_1^{\Omega} \cup \mathcal{E}_2^{\Omega}\cup  \hat{\mathcal{E}}_2^{\Omega} \cup \mathcal{E}_3^{\Omega}\cup \hat{\mathcal{E}}_3^{\Omega},$ where $\mathcal{E}(\Omega)$ is as defined in \cref{def:exchanges}. 
            
Note that $\abs{\mathcal{E}_1} = (k-r)^2 (k-c)^2 , \abs{\mathcal{E}_2} = (k-r)^2, \abs{\hat{\mathcal{E}}_2} = (k-c)^2 , \abs{\mathcal{E}_3} = \binom{k-r}{2}^2, \abs{\hat{\mathcal{E}}_3} = \binom{k-c}{2}^2 .$

Let $\gamma: = \max\set*{  \abs{\det(A_{S\Delta U})} \cdot \abs{\det(A_{T\Delta U})}\given U \in \mathcal{E}(S, T)}.$
By the triangle inequality, $\abs{s_1} \leq \card{\mathcal{E}_1} \gamma.$ A similar inequality holds for $s_2, \hat{s}_2, s_3, \hat{s}_3.$

Consider \cref{eq:generalCase}. Let $M := (k-r)^2 + (k-c)^2 + (r-c)^2$. By the triangle inequality and the above observation,
\begin{align*}
    &M\cdot [\rows{S}, \cols{S}]\cdot [\rows{T}, \cols{T}]\\
    &\leq \gamma(2\abs{\mathcal{E}_1} +  \abs{k-1+c-2r} \cdot \abs{\mathcal{E}_2} +  \abs{k-1+r-2c}\cdot \abs{\hat{\mathcal{E}}_2} + 4(\abs{\mathcal{E}_3} + \abs{\hat{\mathcal{E}}_3} ))\\
    &\leq \gamma\left(((k-r)^2 + (k-c)^2)^2 + 8k ((k-r)^2 + (k-c)^2) \right)\\
    &\leq M(2 k^2 + 8k).
\end{align*}
Since $S\neq T$ so $M > 0$. Dividing both sides by $M$ gives the desired inequality.
\end{proof}

\section{A Crude Approximation Algorithm}\label{sec:crude}

In this section we describe a crude approximation algorithm that can be used to provide the starting point for \cref{alg:localsearch}. We will formally prove \cref{lem:crude}. Our strategy is to appeal to prior results on simpler variants of determinant maximization. Specifically we use the following result of \textcite{Nik15}:\footnote{We remark that the approximation factor of $2^{O(k)}$ is not very important, and one can use simpler and cruder algorithms, such as \cite{approxLargestSimplex}, instead of \cite{Nik15}.}
\begin{theorem}[\cite{Nik15}]\label{thm:nikolov}
	There is a polynomial time algorithm that given a positive semidefinite matrix $B\in \R^{n\times n}$ and $k\geq 0$, outputs a set $S\in \binom{[n]}{k}$ that approximately maximizes $\det(B_{S, S})$. The approximation factor of this algorithm is guaranteed to be $2^{O(k)}$.
\end{theorem}

Using \cref{thm:nikolov}, we will provide an algorithm that constructs $S_0$, a $(n+m)^{O(k)}$-approximation to $\maxdet_k(A)$ in the general case where $k<\min\set{m, n}$.

\begin{proof}[Proof of \cref{lem:crude}]
	
Consider the following  procedure that outputs $S=(\rows{S}, \cols{S})\in \I$:
\begin{enumerate}
    \item Let $B := A A^\intercal \in \R^{m\times m}.$ Note that $B$ is positive semidefinite. Use \cref{thm:nikolov} to pick $\rows{S} \in \binom{[m]}{k}$ that approximately maximizes $\abs{\det (B_{\rows{S},\rows{S}})}$. 
    \item Let $C: = A_{\rows{S},[n]}\in \R^{k\times n}, D:= C^\intercal C \in \R^{n\times n}. $ Use \cref{thm:nikolov} to pick $\cols{S} \in \binom{[n]}{k}$ that approximately maximizes $\abs{\det (D_{\cols{S},\cols{S}})}.$
\end{enumerate}

We claim that for $S=(\rows{S}, \cols{S})$:
\[(n+m)^{O(k)} \cdot \abs{\det(A_{S})}\geq  \max\set*{\abs{\det(A_{T})}\given T\in \I}.\]

Let $T\in \I$ denote the indices of the submatrix with the maximum $k\times k$ subdeterminant. Note that $B_{\rows{S},\rows{S}} = C C^\intercal$. Thus, by the Cauchy-Binet formula,
\begin{equation} \label{ineq:B}
    \begin{split}
        \det(B_{\rows{S},\rows{S}}) &= \sum_{\cols{W} \in \binom{[n]}{k}} \det(C_{[k],\cols{W}}) \det(C^\intercal_{\cols{W}, [k]}) =  \sum_{\cols{W} \in \binom{[n]}{k}} \det(C_{[k],\cols{W}})^2\\ & = \sum_{\cols{W} \in \binom{[n]}{k}}\det(A_{\rows{S}, \cols{W}})^2
    \leq \sum_{\cols{W} \in \binom{[n]}{k}} 2^{O(k)}\cdot \det(A_{\rows{S}, \cols{S}})^2\\ & = n^{O(k)}\cdot \det(A_{\rows{S}, \cols{S}})^2.
    \end{split}
\end{equation}
Similarly, the Cauchy-Binet formula applied to $B_{\rows{T},\rows{T}} = A_{\rows{T},[n]} (A_{\rows{T},[n]})^\intercal $ gives
\begin{equation}
     \det(B_{\rows{T},\rows{T}}) = \sum_{\cols{W} \in \binom{[n]}{k}}\det(A_{\rows{T}, \cols{W}})^2 \geq \det(A_{\rows{T}, \cols{T}})^2
\end{equation}
Thus,
\begin{equation*}
    n^{O(k)} \det(A_{\rows{S}, \cols{S}})^2 \geq 2^{O(k)} \det(B_{\rows{S},\rows{S}}) \geq \det(B_{\rows{T},\rows{T}}) \geq  \det(A_{\rows{T}, \cols{T}})^2,
\end{equation*}
 where the first inequality follows from (\ref{ineq:B}) and the second from definition of $\rows{S}$.
\end{proof}

\section*{Acknowledgements}
We would like to thank Aleksandar Nikolov for initial discussions about general subdeterminant maximization.

\printbibliography
\end{document}